\documentclass[12pt]{article}
\usepackage{amsfonts,amssymb,amsbsy,amsmath,amsthm,mathrsfs,enumerate,verbatim}
 \usepackage[colorlinks=true]{hyperref}

\usepackage{csquotes}

\usepackage[T1]{fontenc}

\usepackage{graphicx}

 \hypersetup{urlcolor=blue, citecolor=red}
\newtheorem{theorem}{Theorem}[section]

\newtheorem{pr}[theorem]{Proposition}

\newtheorem{defi}[theorem]{Definition}

\theoremstyle{definition}

\newcommand{\bel}{\begin{equation} \label}
\newcommand{\ee}{\end{equation}}

\newcommand{\rd}{{\mathbb R}^{2}}

\newcommand{\R}{{\mathbb R}}

\newcommand{\Z}{{\mathbb Z}}

\def\beq{\begin{equation}}
\def\eeq{\end{equation}}
\newcommand{\bea}{\begin{eqnarray}}
\newcommand{\eea}{\end{eqnarray}}
\newcommand{\beas}{\begin{eqnarray*}}
\newcommand{\eeas}{\end{eqnarray*}}

{

\begin{document}

\begin{center}
{\Large \bf The Classical and Quantum Monodromy of the Champagne bottle potential}

\today \\ 
% (Provisional version)

\end{center}

\medskip

\begin{center}
	Quang Sang PHAN
	
	Faculty of Fundamental Sciences,   ~\\
	PHENIKAA University, Yen Nghia, Ha Dong, Hanoi 12116, Vietnam 
	
	E-mail: sang.phanquang@phenikaa-uni.edu.vn
\end{center}

\begin{abstract} 
We illustrate an idea that geometric proprieties have a strong influence to quantum proprieties, through a famous model: the Champagne bottle. 
It is one of the simplest and typical examples of Liouville integrable systems that exhibit a non-trivial classical monodromy. This geometrical invariant perturbs globally the existence of action-angle coordinates on the phrase space. 
However, our work shows a new way to detect the geometric modification on the phrase space by looking at the spectrum of a single operator, that is small non-selfadjoint perturbation of a selfadjoint semiclassical operator accepting  the Hamiltonian of the Champagne bottle as its principal symbol. 
\end{abstract}

\medskip

% {\bf  AMS 2010 Mathematics Subject Classification:} 35R30.\\

\vspace{1cm}

Keywords: Integrable system, Champagne bottle, Non-selfadjoint, Spectral asymptotic

~~\\

\tableofcontents

%%%%%%%%%%%%%%%%%%%%%%%%%%%%%%%%%%%%%%%%%%%%%%%%%%%
%%%%%%%%%%%%%%%%%%%%%%%%%%%%%%%%%%%%%%%%%%%%%%%%%%%

 \fontsize{13}{18}  \selectfont

\section{Introduction}

The aim of this work is to study dynamical problems and the relationship with quantum problems, especially for a famous example-the Champagne bottle, with the help of semiclassical and microlocal analysis, and the spectral asymptotic theory.
We focus on two types of problems, on the one hand, a geometrical invariant-the classical monodromy of Liouville integrable systems, and on the other hand, the spectral study of non-selfadjoint Hamiltonian operators in the semi-classical limit. Such problems are of interest in classical mechanics, and quantum mechanics.

There exist a lot of examples of integrable systems with non-trivial monodromy. Along with the Spherical pendulum (Refs. \cite{Cush88}, and \cite{Gui89}), the Champagne bottle is one of the simplest and typical classical systems that exhibit a non-trivial monodromy. Here the Hamiltonian describes a particle moving in a potential field in the plane ($n=2$) shaped like a Champagne bottle and under the influence of rotationally symmetric potential ($S^1$ symmetric double well). Then the Hamiltonian $H$ including kinetic plus potential has the angular momentum as a conserved quantity, denoted by $J$. Therefore the system is completely integrable. The expressions of $H$ and  $J$ will be explicitly given by \eqref{H} and \eqref{J} in the next section.

This system was mentioned in many works and was studied in Ref. \cite{Bat91} by L. M. Baste. However that work lacked clarity, and especially the absence of the classification of singularities of the system. In fact the system has a unique singularity of focus-focus type, which is the crucial property leading to a non-trivial monodromy.

For this reason, in the first part of this project, we propose reconsidering the classical Champagne bottle, including a detailed classification of singularities (in the sense of Vey- Williamson) and discussing the geometric change of leaves on the phase space (with a remark on focus-focus singularities from Ref. \cite{Z97}) to show the non-trivial monodromy.

On the other hand, it is known upon many existing works that the classical dynamic has a great influence on the quantum aspect.
How can we detect the classical monodromy from quantum systems?

\textbf{The first answer} for this question is the quantum monodromy. It is an obstruction to the existence of a
global lattice structure for the joint spectrum of system of selfadjoint semiclassical operators that commute.
In the semiclassical limit, the quantum monodromy allows to recover the classical monodromy of the underlying classical system

The quantum monodromy was detected a long time ago for the Hamiltonian of the spherical pendulum, see Refs. \cite{Cush88}, and \cite{Gui89}. Then it was completely defined by S. V\~{u} Ng\d{o}c in Ref. \cite{Vu-Ngoc99}.

In the framework of this work, we will give explicitly the quantum system of the Champagne bottle, obtained by the quantifications of $H$ and  $J$, denoted by $\widehat{H}$ and $\widehat{J}$. Using a particular result form Ref. \cite{Vu-Ngoc99}, we can conclude that the non-trivial monodromy for the Champagne bottle can be given by the matrix
 \begin{equation}  \label{m}
 \mathcal{M}= \left[
                                      \begin{array}{cc}
                                        1 & 0 \\
                                        1 & 1 \\
                                      \end{array}
                                    \right]
\end{equation}
 in the group $ GL(2, \mathbb Z)$, modulo conjugation.

\textbf{The second answer} for the previous question is the spectral monodromy, defined in Refs. \cite{QS14}, and \cite{QS18}.
This invariant is directly defined  from the spectrum of a single non-selfadjoint semiclassical operator with two degrees of freedom, as the obstruction to the existence of a globally lattice structure for the spectrum, in the semiclassical limit.
The new idea here is that the spectral monodromy allows to detect the modification of action-angle variables from only one spectrum.

For the Champagne bottle, we will study the spectral monodromy of non-selfadjoint semiclassical operators, that are
small non-selfadjoint perturbations of $\widehat{H}$.

In summarizing, our work focus on the following goals:
\begin{enumerate} [(1)]
    \item Classify the singularities of the Champagne bottle, and show the non-trivial monodromy.
    \item Give explicitly the quantum system $(\widehat{H}, \widehat{J})$ for the Champagne bottle and its quantum monodromy.
    \item Study the spectral monodromy of small non-selfadjoint perturbations of $\widehat{H}$, and clarify that this monodromy allows to recover the classical monodromy of the Champagne bottle.
\end{enumerate}
% We would like to note that numeric calculation is very expected to clarify these goals.

\section{The classical Champagne bottle}

First we recall that a Hamiltonian $H$ on a $2n$-dimensional symplectic manifold (called the phase space) is said to be completely integrable if there exist $n$ independent almost everywhere smooth functions $f_1=H,f_2,\cdots, f_n$ that pairwise commute with respect to the Poisson bracket induced by the symplectic form of the manifold.
It is known from Action-angle Theorem (by Liouville, Mineur and Arnold,  Ref \cite{Duis80}) that each compact regular leaf of the momentum map $F=(f_1,f_2,\cdots, f_n)$ is a $n-$Lagrangian torus, on which the Hamiltonian flow of $H$ is quasi-periodic. Moreover, there exist local canonical coordinates $(x, \xi)=(x_1,x_2,\cdots, x_n, \xi_1, \xi_2,\cdots,  \xi_n )$ called action-angle variables, on a neighborhood of each torus so that the Hamiltonian depends only on the action variable $\xi$.
However, it is well known from Duistermaat's Ref. \cite{Duis80} that there exists an obstruction to the existence of globally action-angle coordinates, called the classical monodromy.

\subsection{The singularities of the Champagne bottle}
Let the phase space be the cotangent bundle $M=T^*\rd$ and we denote $ (x,\xi)=(x_1,x_2,\xi_1, \xi_2)$ the variables on $M$.
We suppose that a particle (but without loss of generality, assuming be of mass $m=1$) moves in the plane $\rd$ under the rotationally symmetric potential
$$V(x_1,x_2)=V(r)=r^4-r^2, \ \textrm{where} \ r= \sqrt{x_1^2+ x_2^2}, \ (x_1,x_2) \in \rd. $$ 

\begin{figure}[!h]
	\begin{center}
		\includegraphics [width=0.6\textwidth]{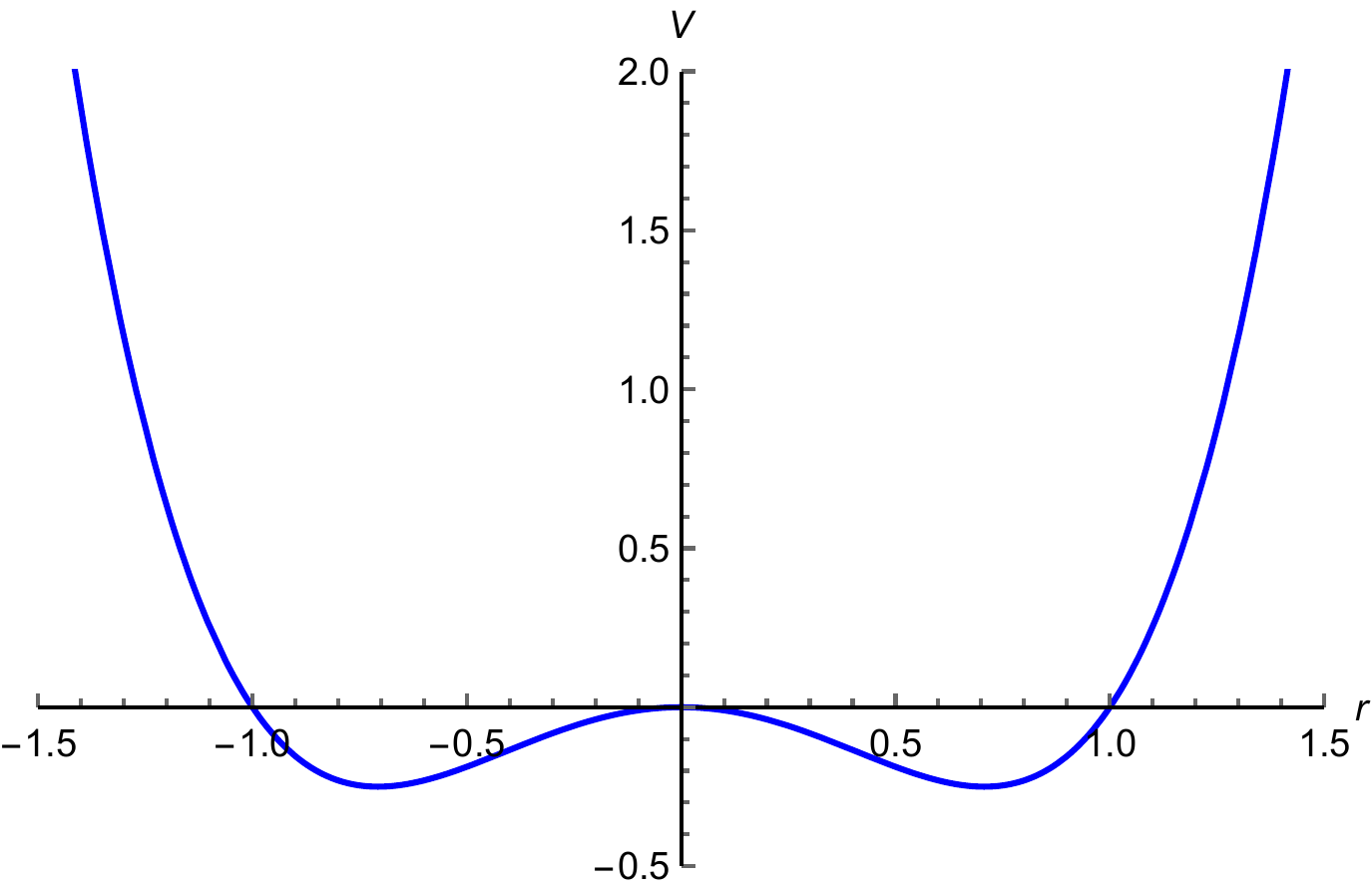}
		\caption{The potential of the Champagne bottle}
	\end{center}
	\label{f1}
\end{figure}

The Hamiltonian of the Champagne bottle is given by
 \begin{eqnarray}   \label{H}
H(x,\xi) &=&\frac{1}{2} (\xi_1^2+ \xi_2^2)+ V(x_1,x_2) 
 \\
 &= &  \frac{1}{2} (\xi_1^2+ \xi_2^2) + (x_1^2+ x_2^2)^2-(x_1^2+ x_2^2).    \nonumber
 \end{eqnarray}
 
It is clear that the system is rotationally invariant. Then the angular momentum is conserved and therefore Poisson commutes with $H$. Hence the system $H$ is completely integrable.
Here the angular momentum is
 \begin{equation} \label{J} J(x,\xi)= x_1 \xi_2-  x_2 \xi_1.   \end{equation}
Then the momentum map for the Champagne bottle is
$$F=(H, J): T^*\rd \longrightarrow \rd. $$

We recall that a singular point of $F$ means a point at which the differential $dF$ hasn't a maximum rank. 
Using the polar coordinates to rewrite $H$ and $J$, we can give the range of the momentum map $F$ and compute explicitly the singularities of $F$. 

For $r>0$, we introduce the polar coordinates $(r, \theta)$ and we denote $(\xi_1', \xi_2')$ the conjugate variables associated to $(r, \theta)$. 

% 	$$  \left\{     \begin{array}{l}  x_1=r \cos \theta \\  x_2=r \sin \theta  \end{array}  \right.,  $$

Now we look at the differential $dF (x, \xi)$ on the tangent space to find the points where $F$ has rank less than two. 

\begin{pr} 
	\begin{enumerate} [(1)] 
		\item  $dH=0$ if and only if $x_1 = x_2=\xi_1= \xi_2=0$, or $\xi_1= \xi_2=0, \ r= \frac{1}{\sqrt{2}}$.
		\item $dJ=0$ if and only if $x_1 = x_2=\xi_1= \xi_2=0$. 
		\item In the polar coordinates, the singularities of rank 1 of $F$ are given by 
		\begin{equation} \xi_1'=0, \xi_2'^2=4r^6-2r^4 \  \textrm{for} \  r \geq  \frac{1}{\sqrt{2}}. 
		\end{equation}
	\end{enumerate}
\end{pr} 

\begin{proof}
% [Proof of Proposition ]

When $x_1 = x_2=0$, 
$$dH= \xi_1 d \xi_1+ \xi_2 d \xi_2,  $$
and 
$$dJ=  \xi_2 d x_1- \xi_1 d x_2. $$
Then $dH \wedge  dJ= 0 \Leftrightarrow  \xi_1= \xi_2=0    $, and $(0,0)$ is a critical value of $F$.

In the polar coordinates $(r, \theta)$ where $r >0$, $H$ and $J$ become 
 \begin{eqnarray}   
H &=&  \frac{1}{2} (\xi_1'^2+ \frac{1}{r^2} \xi_2'^2) + r^4-r^2,  \nonumber
\\
J&= &  \xi_2'.  \label{J2}  
\end{eqnarray}

Then 
 \begin{eqnarray}   
dH &=&  \xi_1'd\xi_1'+       \frac{1}{r^2} \xi_2'  d\xi_2' + (4r^3-2r-\frac{ \xi_2'^2  }{r^3}) dr, \nonumber
\\
dJ&= & d \xi_2'.      \nonumber
\end{eqnarray} 
The form $dH \wedge dJ$ has rank 1 if and only if $dH= \lambda dJ$, $\lambda \in \R$, and this happens when 
\begin{equation} \label{cri pt} \xi_1'=0, \xi_2'^2=4r^6-2r^4 \  \textrm{for} \  r \geq  \frac{1}{\sqrt{2}}. 
\end{equation}
In this case, $dh= \lambda dJ$, for 
\begin{equation} \label{lbd}   
\lambda = \pm \sqrt{4r^2-2}, \  r \geq \frac{1}{\sqrt{2}}, 
  \end{equation}
and the critical values of rank 1 of $F$ (see Fig. 2) is the curve parametrized by 

  \begin{equation} \label{critical values}  
   (H, J)= ( \pm \sqrt{4r^6-2r^4}, 3r^4-2r^2), \  r \geq \frac{1}{\sqrt{2}}.
  \end{equation}

\end{proof}

\begin{figure}[!h]
	\begin{center}
		\includegraphics [width=0.45 \textwidth]{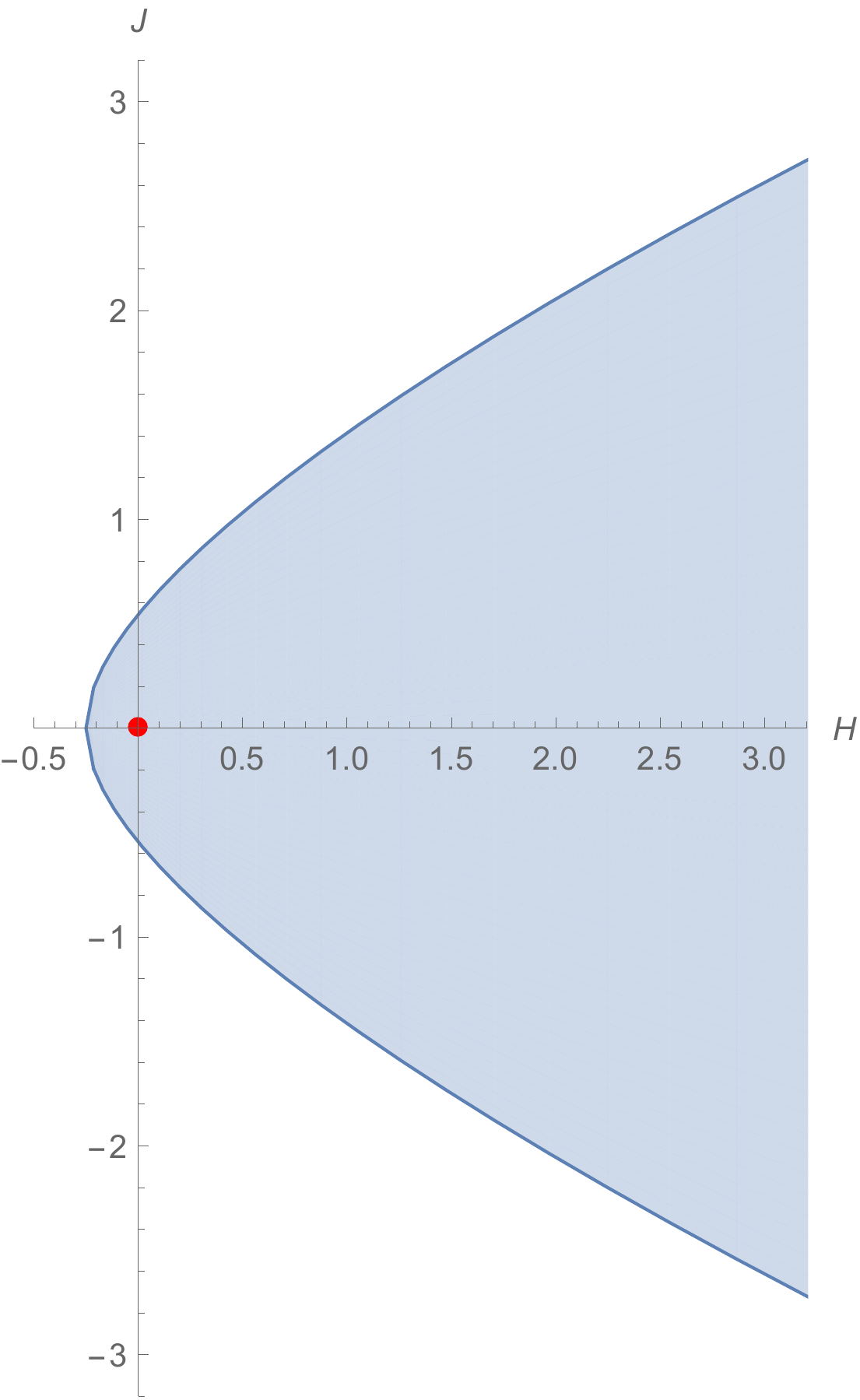}
		\caption{The critical values of the Champagne bottle}
	\end{center}
	\label{f1}
\end{figure}
\begin{pr} The image of $F$ is the connected domain, containing the origin, bounded by the curve parametrized by \eqref{critical values}. This curve is the set of critical values of $F$ of rank 1, the point $(0,0)$ is the image of critical points of rank $0$. (see Fig. 2)
\end{pr} 

\begin{proof}
It is clear that the range of $H$ is $[- \frac{1}{4}, + \infty [  $. 
Since $H$ is a continuous function, so the level sets $H=E$, for any  
$E \geq - \frac{1}{4}$, is nonempty and compact. 

For $E=0$, as we see in the above paragraph, the single critical point of rank 2 of $F$ is $(0,0,0,0)$ corresponding to the critical value $(H=0, J=0)$.
 
For other values of $E$, the hypersurface $H^{-1}(E)$ is compact, on which the function $J$ is so bounded and its image is a close interval. As the extreme points of $J$ on these hypersurface are given by the equation $dH=\lambda dJ$, so they are of the form \eqref{cri pt}. The values of $H$ and $J$ at these points are given by \eqref{critical values}. For a given value $E \geq - \frac{1}{4}$, there exactly a single value $ r \geq \frac{1}{\sqrt{2}}$ and so two opposite values for $J$ as in \eqref{critical values}. 
\end{proof}

\subsection{Classification of the singularities}

In this part, we classify the singularities of $F$ and describe the geometry of the leaves of $F$ in phrase space. Here the type of singularities is classified  in the sense of Vey- Williamson. We recall  the definition for that classification.

Let  $F= (f_1, f_2, \ldots, f_n )$ be an integrable system.  

\begin{defi} [Ref. \cite{Vey78}] 
	We say that a point $m$ is a fixed point (or of corank $0$) of $F$ if $F$ has a non-differential at $m$.	
	Moreover, this fixed point $m$ is called non-degenerate (in sense of Vey-Williamson) if the Hessians $\mathcal{H}(f_j) $  generates a
	Cartan subalgebra of the Lie algebra $\mathcal{Q}(2n)$ of quadratic forms on $T_mM$ provided the Poisson bracket. 
\end{defi}

% lemma 3.3.6 pag 77 cite condition for a Cartan subalgebra. 

In general case, if $m$ is a critical point of corank $r$ (that means $dF(m)$ has rank $n-r$). We can suppose  for example that $df_1(m), df_2(m), \ldots, df_{n-r}(m)$ are linearly independent. Then the action of $f_1(m), f_2(m), \ldots, f_{n-r}(m)$ reduces a symplectic manifold $\sum$ of $M$, and we say that $m$ is non-degenerate (or also transversely non-degenerate) if the restriction of the system on $\sum$  (means $f_1(m), f_2(m), \ldots, f_{n-r}(m)$) admits $m$ as fixed point  non-degenerate, as in the above definition.  In this case the theorem of Eliasson (see Refs. \cite{Eli84}, and \cite{Eli90}) assures that the system admits a linear model on $T^* \R^{n-r} \times T_m \Sigma$ of the form
 \begin{equation}  \label{Q} 
 Q= (\xi_1, \ldots,\xi_{n-r}, q_1, \ldots, q_r ),
   \end{equation}
where $\xi_j$ are the canonical coordinates of $T^* \R^{n-r}$, and the $q_j$ form a basic for the Cartan subalgebra indicated in the above definition. 

\begin{theorem}[Eliasson, Ref. \cite{Eli84}] 
	The non-degenerate critical points are linearizable: there exists local symplectomorphism $\kappa$ from $ (R^{2n},0) $ into a neighborhood of $(M,m)$ such that
			$$\kappa ^* F=\varphi (Q),$$
	where $Q$ is the linear model as in \eqref{Q}, and  $\varphi:  (R^n,0) \rightarrow  (R^n,F(m)) $ is a local diffeomorphism. 	
\end{theorem}

\begin{theorem}[Williamson, Ref. \cite{Wi36}] 
	Let $\mathfrak{c}$ be a Cartan real subalgebra of $\mathcal{Q}(2n)$. There exists canonic coordinates $(x_1, \ldots, x_n, \xi_1, \ldots,\xi_n)$ on $\R^{2n}$, and a \enquote{standard basic} $ (q_1, \ldots, q_n)$ for $\mathfrak{c}$ such that each $q_j$ has one of three following forms: 
	\begin{enumerate} [(1)]
		\item  $q_j= x_j \xi_j$ (hyperbolic singularity)
		\item  $q_j=  \frac{1}{2}( x_j^2+  \xi_j^2)$ (elliptic singularity)
		\item  $q_j= x_j \xi_j+ x_{j+1} \xi_{j+1} $ (focus-focus singularity)
	\end{enumerate}	
\end{theorem}

Note that according to Ref. \cite{Eli84}, the singular points are classified by the number of $q_j$ of the above types in a standard basic of $\mathfrak{c}$.

We come back now to the Champagne bottle. We will show that the system admit a unique focus-focus singularity, which causes a non-trivial monodromy.
For this we need look at the Hessians of $H$ and $J$ at their singularities. 

First, for the fixed point $M_0(0,0,0,0)$, we know from the previous section that it is the only critical point of the critical level $\Lambda_0= F^{-1}(0)$. The Hessian matrices of $H$ and $J$ at this point are

\begin{equation*}
d^2  (H)(0,0,0,0) = 
\left[
\begin{array}{cccc}  
-2& 0 & 0 & 0 \\
 0& -2 & 0 & 0\\
 0& 0 & 1 & 0 \\
 0& 0 & 0 & 1\\
\end{array}
\right], 
\end{equation*}

\begin{equation*}
d^2(J)(0,0,0,0) = 
\left[
\begin{array}{cccc}  
0& 0 & 0 & -1 \\
0& 0 & 1 & 0\\
0& 1 & 0 & 0 \\
-1& 0 & 0 & 0\\
\end{array}
\right].
\end{equation*}

The Hessians of $H$ and $J$ at $M_0$ are the quadratic forms associated to the above matrices:
$$ \mathcal{H}(H)  (M_0)= \xi_1^2-2x_1^2+ \xi_2^2-2x_2^2  , $$ 
$$ \mathcal{H}(J)(M_0)= 2(\xi_1x_2-\xi_2 x_1).  $$ 
We can show that in some local symplectic coordinates $(x,y, \xi, \eta)$ of $M_0$, the Hessians $ \mathcal{H}(H)$ and $ \mathcal{H}(J)$ generate a 2-dimensional subalgebra of the algebra $\mathcal{Q}(4)$ of quadratic forms in $(x,y, \xi, \eta)$ under Poisson bracket that admits the following basis $(q_1, q_2)$: 
  
   $$q_1= x \xi+ y \eta, $$
   $$q_2= x\eta -y \xi.  $$
Actually, to do this, we just use the canonical substitutions 
$$ x_1= \frac{1}{\sqrt[4]{2}}x_1', \ \xi_1= \sqrt[4]{2} \xi_1', \  x_2= \frac{1}{\sqrt[4]{2}}x_2', \ \xi_2= \sqrt[4]{2} \xi_2' $$
and then 
$$ x_1'= \frac{1}{\sqrt{2}}(x- \xi), \ \xi_1'= \frac{1}{\sqrt{2}}(x+\xi), \  x_2'= \frac{1}{\sqrt{2}}(y- \eta), \ \xi_2'= \frac{1}{\sqrt{2}}(y+ \eta).$$

Such a singularity $M_0$ is called a non degenerate focus-focus singularity.  
This point is an equilibrium point for both the dynamics of $H$ and $J$. Moreover, by looking at the stability of the dynamics of $q_1$ and $ q_2$ near the origin and therefore the ones of $H$ and $J$, we can conclude that this singularity $M_0$ is unstable for $H$, but contrariwise is stable for $J$. 

For a linear combination of $J$ and $H$, $M_0$ is a hyperbolic point whose close trajectories move away in spiraling on the unstable manifold and converge in spiraling on the stale manifold, see Fig.

% This leave is compact, connected. 

 For the singular points of rank 1, whose image is on the curve \eqref{critical values} : they are of type transversely elliptic.  
 There exists a linear combination of $J$ and $H$ whose Hessian at such a singular point is a harmonic oscillator of one dimension. In fact, such a linear combination is of the form $H+ \lambda J$ for a value $ \lambda $ given by \eqref{lbd}. They are of type transversely elliptic singularities.

 For example, the singular points corresponding to the critical value $(0, - \frac{1}{4})$ are of the form $(\frac{1}{\sqrt{2}}, \theta, x_1'=0, x_2'=0) $ in the reduced space (that is $r= \frac{1}{\sqrt{2}}, \ \xi_1= \xi_2=0$  in the initial phrase space). The Hessian matrices of $H$ and $J$ at such a point are
 
 \begin{equation*}
 d^2  (H) = 
 \left[
 \begin{array}{cccc}  
 4& 0 & 0 & 0 \\
 0& 0 & 0 & 0\\
 0& 0 & 1 & 0 \\
 0& 0 & 0 & \frac{1}{2}  \\
 \end{array}
 \right],  \   d^2(J) = 0. 
 \end{equation*}
 Then the corresponding Hessians of $H$ and $J$  are
 $$ \mathcal{H}(H) = 4 r^2+ \xi_1'^2+ \frac{1}{2} \xi_2'^2 , \  \mathcal{H}(J)=0.  $$ 
These Hessians generate a 1-dimensional subalgebra of the algebra $\mathcal{Q}(4)$ of quadratic forms with a basis of a single harmonic oscillator. 
These singular points are stable equilibrium points for the dynamic of $H$. The Hessian of $H$ at such a singular point is a harmonic oscillator of one dimension.

\subsection{Fibration of the Champagne bottle}

In this section, we are studying the semi-global aspect of the Champagne bottle in talking about the geometry of the leaves of $F$ and considering $F$ as a singular Lagrangian foliation. 
For any value $c$ in the image of $F$, we denote the leaf  $\Lambda_c=F^{-1}(c) $. It is clear that $\Lambda_c $ is an invariant connected compact subset on the phrase space. 

Let us denote $U_r$ the set of regular values of $F$. 
If $c \in U_r$ is a regular value of $F$, then $\Lambda_c $ is always an invariant Lagrangian torus of 2-dimension (a Liouville torus) by  Action-angle theorem of Liouville-Arnold-Mineur, see Ref. \cite{Duis80}]. The space of regular leaves of F is so foliated by Liouville invariant tori. 
 Moreover, there exists a locally canonic transformation on a neighborhood of $\Lambda_c $, called angle-action coordinates. 

If $c $ is a critical value of rank 1, (i.e. a point on the curve \eqref{critical values}) then the critical fiber $\Lambda_c $ is a circle. 
In fact, let $m  \in \Lambda_c $ be a singular point of rank 1 (that satisfies \eqref{cri pt}). 
Then an application of Eliasson's theorem deduces that on a neighborhood of $m$, the system is locally symplectically reduced to the model
$$Q= (\xi, q), \ q=  \frac{1}{2}( y^2+  \eta^2), \ \textrm{on} \  T^* \mathbb{T}_{(x, \xi)}\times   T^* \mathbb{R}_{(y, \eta)}.$$
Hence the orbit of $m$ under the joint flot of $F$ (which is symplectically equivalent to the orbit of $0$ by the joint flot of $Q$) is a torus of dimension $1$. 
Moreover, the close fibers are 2-dimensional Liouville tori which can be seen as the product of the fiber by tori of dimension $1$ of transversal elliptic foliation given by $q$. 

The last case, if $c $ is the origin $(0,0)$, then $\Lambda_c $ is a pinched torus (a torus whose one of two cycles degenerates into a point). As we are known that $\Lambda_0 $ admit the point $M_0(0,0,0,0)$ as the only singularity, which is of focus-focus type.  
The Eliasson's theorem deduces that near this point, the system is symplectically conjugated to 
$$Q= (q_1, q_2), \  q_1= x \xi+ y \eta, q_2= x\eta -y \xi, \ \textrm{on} \    T^* \mathbb{R}^2_{(x,y,\xi \eta)}.$$
We remark that the Hamiltonian flot of $q_2$ is periodic of period $2 \pi$. Then $\Lambda_c $ is a pinched torus and the close fibers of $\Lambda_c $ are standard 2-dimensional Liouville tori $\mathbb{T}^2$. 

\begin{figure}[!h]
	\begin{center}
		\includegraphics [width=0.9 \textwidth]{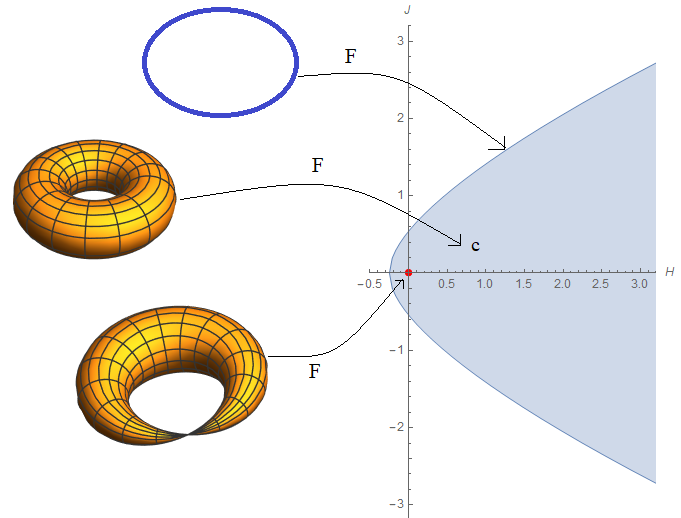}
		\caption{The fibration of the Champagne bottle}
	\end{center}
	\label{f1}
\end{figure}

We see clearly a change in the topology of leaves of $F$ near the focus-focus singularity. This singularity is the centre of a non-trivial monodromy as we shall explain in the next section.

\subsection{The classical monodromy of the Champagne bottle}

Let $U \subseteq U_r$ be a punctured open subset around $(0,0)$. The classical monodromy is an obstruction to the existence of global angle-action coordinates for the integrable system $F$, see Ref. \cite{Duis80}.  It is defined as the triviality property of the $\mathbb Z^n-$ bundle 
$$H_1(\Lambda_c, \mathbb Z) \rightarrow c \in U,$$
here $H_1(\Lambda_c, \mathbb Z)$ is the first homology group of the Liouville torus $\Lambda_c$. 

It is clear that for each $c \in U$, $H_1(\Lambda_c, \mathbb Z) \sim \pi_1 (\Lambda_c)$. Therefore the triviality of the previous bundle is equivalent to the one of the following bundle 
		$$\pi_1 (\Lambda_c) \rightarrow c \in U.$$

Suppose that $\{ U^j \}_{j \in \mathcal{J} }$ is an arbitrary open finite covering of $U$ (the diameter of each of the sets $U^j$ is sufficiently small).
Here $\mathcal{J}$ is a finite index set. We assume that $F^{-1}(U) \subseteq M$ can be covered by a finite covering of angle-action charts $ \{(V^j, \kappa ^j ) \} _ {j \in \mathcal{J} }$, $ V^j= F^{-1} (U_j) $.
Then the transition maps between trivializations of the previous bundle are given by 
$$\{ {}^t \big ( M_{ij}^{cl} \big ) ^{-1} \} _{i, j \in \mathcal{J}}, $$
where  $M _{ij}^{cl}=  d( ({\varphi ^i}) ^{-1} \circ \varphi ^j  )  $
are a locally constant matrices of $ GL(2, \mathbb Z)$ defined on the nonempty overlaps $V^i \cap V^j$, $i, j \in \mathcal{J}$. 

It is well known that the bundle has a nontrivial monodromy.
A result from Ref. \cite{Vu-Ngoc99} (see also Ref. \cite{Bat91}, ) shows that the monodromy in the case of focus-focus singularity is always non-trivial and given by the matrix  
$$ \mathcal{M}= \left[
\begin{array}{cc}
1 & 0 \\
1 & 1 \\
\end{array}
\right]
$$ 
in the group $ GL(2, \mathbb Z)$, modulo conjugation.

\section{The quantum Champagne bottle}

\subsection{Weyl-quantization}  \label{sec2.1}

We work with pseudodifferential operators obtained by the $h-$Weyl-quantization of a standard space of symbols on $T^*M =\mathbb R^{2n}_{(x,\xi)}$, here $M= \mathbb R^n$ or a compact manifold of $n$ dimensions, $n\geq 1$. We denote $\sigma $ the standard symplectic $2-$form on $T^*M$.

In the following we represent the quantization in the case $M= \mathbb R^n$. In the manifold case, the quantization is suitably introduced.
We refer to Refs. \cite{Dimas99}, and \cite{Shubin01} for the theory of pseudodifferential operators.

% In the following, we introduce classical operators on $M= \mathbb R^n$.

\begin{defi} \label{fonc ord}
	A function $m: \mathbb R^{2n} \rightarrow (0, + \infty)$ is called an order function
	if there are constants  $C,N >0$ such that
	$$m(X)  \leq C \langle X-Y\rangle^{ N} m(Y), \forall X,Y \in \mathbb R^{2n},$$
	with notation $\langle Z\rangle= (1+ |Z|^2)^{1/2}$ for $Z \in \mathbb R^{2n}$.
\end{defi}

% One use often the order function  $m(Z) \equiv 1$ or $$m(Z)= \langle Z \rangle ^{l/2}= (1 + |Z|^2 )^{l/2},$$ with a given constant $l \in \mathbb R $.

\begin{defi}
	Let $m$ be an order function and $k \in \mathbb R$, we define classes of symbols of $h$-order $k$, denoted by $S^k(m)$ (families of functions), of $(a(\cdot;h))_{h \in (0,1]}$ on $\mathbb R^{2n}_{(x,\xi)}$ by
	\begin{equation}
	S^k(m)= \{ a \in C^\infty (\mathbb R^{2n})
	\mid  \forall \alpha \in \mathbb N ^{2n}, \quad |\partial^\alpha a | \leq  C_\alpha h^k m \} ,
	\end{equation}
	for some constant $C _\alpha >0$, uniformly in $h \in (0,1]$. \\
	A symbol is called $\mathcal O(h^\infty)$ if it's in $\cap _{k \in \mathbb R } S^k(m):= S^{\infty}(m) $.
\end{defi}

Then $ \Psi^k(m)(M)$ denotes the set of all (in general unbounded) linear operators  $A_h$ of order  $k$  on $L^2(\mathbb R^n)$, obtained from the $h-$Weyl-quantization of symbols $a(\cdot;h) \in S^k(m) $ by the integral:
\begin{equation} \label{symbole de W}
(A_h u)(x)=(Op^w_h (a) u)(x)= \frac{1}{(2 \pi h)^n}
\int_{ \mathbb R^{2n}} e^{\frac{i}{h}(x-y)\xi}
a(\frac{x+y}{2},\xi;h) u(y) dy d\xi.
\end{equation}

In this work, we always assume that symbols admit classical asymptotic expansions in integer powers of $h$.
The leading term in this expansion is called the principal symbol of operators.

\subsection{The quantum monodromy of the Champagne bottle}     \label{3.2}

By quantization for $H$ and $J$, we get the corresponding semiclassical operators acting on $L^2(\rd)$:
\begin{equation} \label{qH} \widehat{H}=-\frac{h^2}{2} \Delta +(x_1^2+ x_2^2)^2-(x_1^2+ x_2^2),\end{equation}
where $ \Delta$ denotes the Laplace operator, and the remain term assigns the multiplication operator by itself, and
\begin{equation} \label{qJ} \widehat{J}=\frac{h}{i} \big( x_1 \frac{\partial}{x_2}- x_2 \frac{\partial}{x_1} \big).\end{equation}

Then $\widehat{H}$ and $\widehat{J}$ commute.

We will give a description for the joint spectrum of the integrable quantum system $(\widehat{H}, \widehat{J})$.
Here the spectral theory will be used to prove that  $\widehat{H}$ is a semi-bounded, essentially selfadjoint operator with compact
resolvent. 
Therefore the spectrum of $\widehat{H}$ is purely discrete, formed of eigenvalues tending to infinity. The eigenspaces of $\widehat{H}$ are all of finite dimension contained in $C^{\infty}(\R^2)$. The operator $\widehat{J}$ acts on each eigenspace as a Hermitian matrix. Therefore we can diagonalize simultaneously $\widehat{H}$ and $\widehat{J}$ on a completely orthonormal basis $\mathcal{B}$ of $L^2(\R^2)$.

Moreover, since the expression \eqref{J2} of $J$  in the polar coordinates, we have 
$$\widehat{J}= \frac{h}{i}\frac{\partial}{\partial\theta}.$$
So that the eigenvalues of $\widehat{J}$ are integers. However we don't know explicitly the spectrum of $\widehat{H}$.
For each $u \in \mathcal{B} $, we have
\begin{equation}\label{jsp}
\widehat{H} u = \lambda_H^{(u)}u, \  \widehat{J} u = \lambda_J^{(u)} u, \ \textrm{with} \  \lambda_J^{(u)} \in \Z.
\end{equation}
Hence the joint spectrum of the integrable quantum system $(\widehat{H}, \widehat{J})$, that consists of pairs $(\lambda_H^{(u)}, \lambda_J^{(u)})$, $u \in \mathcal{B} $, forms a discrete lattice of $\R^2$, included fully in the image of the $F$.

The quantum monodromy of such a lattice, defined by S. Vu Ngoc in Ref. \cite{Vu-Ngoc99}, is an invariant non-trivial given by the matrix $ \mathcal{M}$ in \eqref{m} given in the group $ GL(2, \mathbb Z)$, modulo conjugation.

\begin{figure}[!h]  
	\begin{center}
		\includegraphics [width=0.8\textwidth]{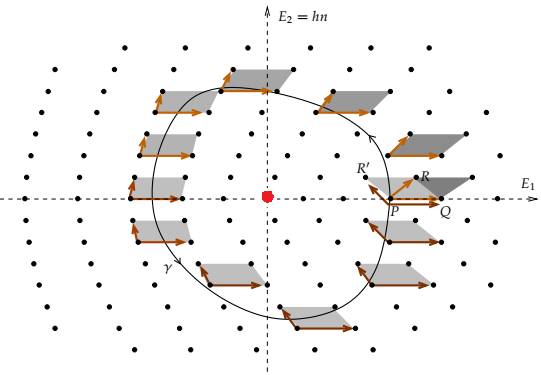}
		\caption{The quantum monodromy (Image from \cite{Vu-Ngoc06}) }
	\end{center}
	\label{fg2}
\end{figure}
Intuitively, if we realize the parallel transport on the lattice of a basic rectangle from a fixed position along a some closed loop in the domain of regular values of $F$ around the origin $(H,J)=(0,0)$ and then return to the starting position, then the initial rectangle becomes a different rectangle (see Fig. 4). This means the existence of non-trivial monodromy.

\subsection{The spectral monodromy of the Champagne bottle}

There exits another way of detecting the monodromy of a completely classical integrable system like the Hamiltonian of the Champagne bottle, by looking at the spectrum of a single operator that is small non-selfadjoint perturbation of a selfadjoint semiclassical operator with two degrees of freedom,  Refs.  \cite{QS14}, and \cite{QS18}. 

The general spectral asymptotic theory, see Ref. \cite{Hitrik07}, allow to describe locally the spectrum of such a perturbation as discrete latices. However with regard to the global problem, there exits a combinatorial invariant, called the spectral monodromy, that obstructs globally the existence of lattice structure of the spectrum, in the semiclassical limit. Moreover this monodromy allows to recover the classical monodromy of the underlying completely integrable system. 
We refer to Ref.  \cite{QS18} for a detailed construction for the spectral monodromy. In the following, first we present a simple case to illustrate a connection between the known monodromy of the Champagne bottle and a monodromy that can be defined from the spectrum, and then  a brief description for the spectral monodromy in a more complex case.

%%%%%%%%%%%%%%%%%%%%%%%%%%

We start with a simple case in considering the operator

\begin{equation}  \label{op1}
 P_\varepsilon=\widehat{H}+ i \varepsilon\widehat{J},
 \end{equation}
where $\widehat{H}$ and $\widehat{J}$ given in Sec. \ref{3.2}, and  $\varepsilon$ is a small real parameter. 
~\\
Then since \eqref{jsp}, it is clear that the spectrum of $ P_\varepsilon$ is the discrete lattice
$$   \sigma (P_\varepsilon) = \{\lambda_H^{(u)} + i \varepsilon \lambda_J^{(u)}, u \in \mathcal{B} \}.$$
We identify $\mathbb C$ with $\mathbb R^2$ and let  $\chi$ be the map
\begin{eqnarray}   \label{chi}
\chi :  \mathbb R^2 \ni u= (u_1,u_2) \mapsto \chi_u &=& (u_1, \varepsilon u_2) \in \mathbb R^2
\\
&\cong & u_1+i \varepsilon u_2 \in \mathbb C.    \nonumber
\end{eqnarray}
Then $\chi^{-1} (\sigma (P_\varepsilon))$ is clearly the joint spectrum of the quantum system $(\widehat{H}, \widehat{I})$. 
Hence a monodromy can be defined for the spectrum of $ P_\varepsilon$ as the quantum monodromy, introduced in Sec. \ref{3.2}.

With the help of a numerical calculation we can explicitly give the spectrum of $P_\varepsilon$. 
Using the software Mathematica for $h=0.1$ and $\varepsilon= \sqrt{h}$,
Fig. 5 shows a lattice structure of $\chi^{-1} (\sigma (P_\varepsilon))$.  
 It is very nice that this discrete lattice is fully included in the image of $F$ (see also Fig. 2).

\begin{figure}[!h]
	\begin{center}
		\includegraphics [width=0.4\textwidth]{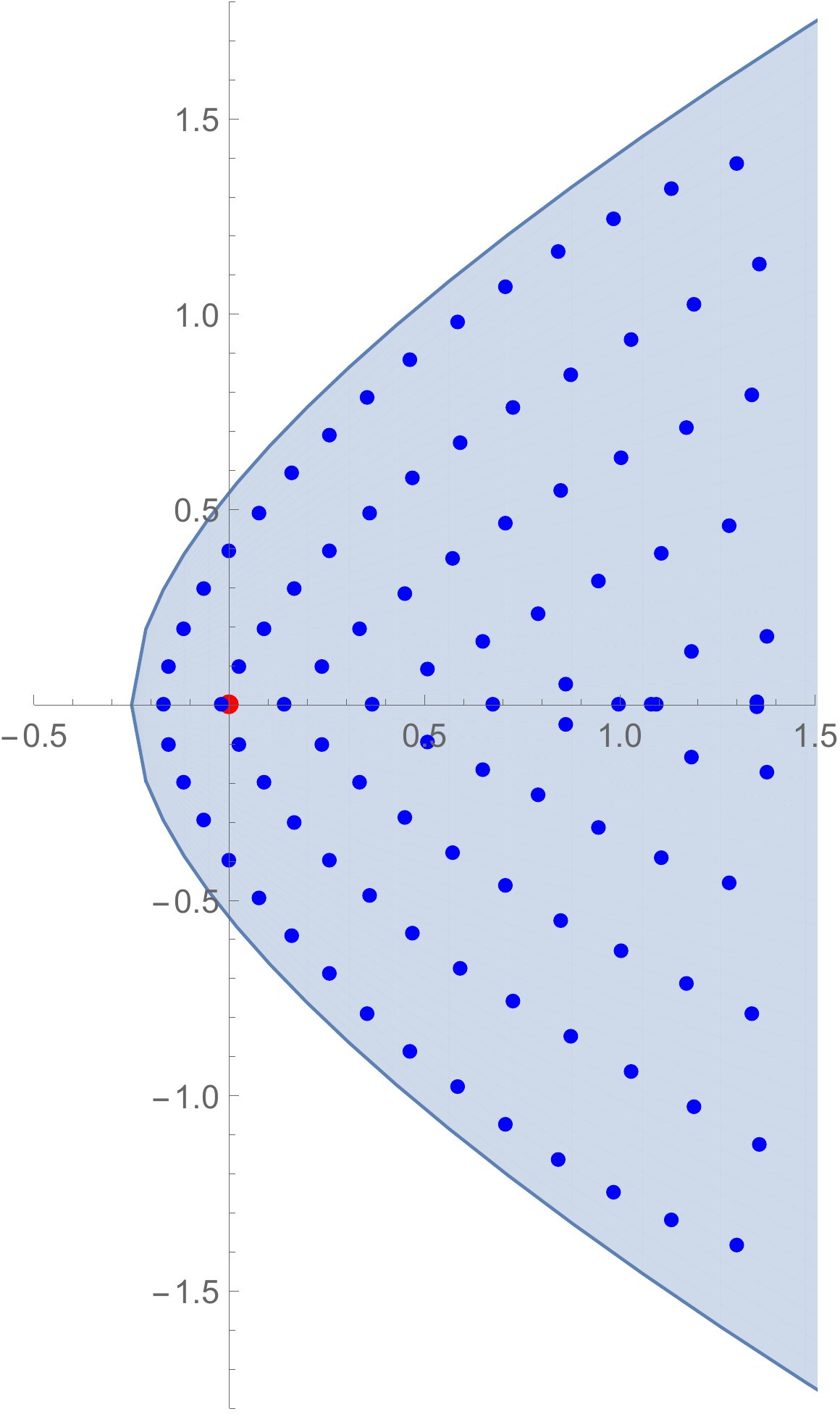}
		\caption{The spectrum of the Champagne bottle through $\chi^{-1}$ }
	\end{center}
	\label{f1}
\end{figure}

%%%%%%%%%%%%%%%%%%%%%%%%%%%%%%%%%%%%%%%%%%%%%%%%%%%%%%%%%%%%%%%%%%%

In the more complex case, we consider a more complex perturbation of  $\widehat{H}$ than the operator given in \eqref{op1}.
In fact we defined in Ref. \cite{QS18} a combinatorial invariant, called spectral monodromy, directly from the spectrum of small non-selfadjoint perturbations of a selfadjoint pseudodifferential operator with two degrees of freedom, in the semiclassical limit, assuming that the principal symbol of the selfadjoint pseudodifferential operator is a completely integrable system. 

Specifically, let  $\varepsilon$ be a small parameter assumed to depend on the classical parameter $h$ and in the regime $h \ll \varepsilon = \mathcal{O}(h^\delta)$, with $0< \delta <1 $. The operator of interest is of order $0$ and of the form
\begin{equation} \label{tt}
P_{\varepsilon}= P(x,hD_x,\varepsilon; h ), \end{equation}
depending smoothly on $\varepsilon$ on a neighborhood of $(0, \R)$ such that the unperturbed operator $P:= P_{\varepsilon=0}$ is formally selfadjoint whose pricipal symbol $p$ is completely integrable. 

Let $p_\varepsilon$ be the principal symbol of $P_{\varepsilon}$ and then $q=\frac{1}{i}(\frac{\partial p_\varepsilon}{\partial \varepsilon})_{\varepsilon
	=0}$, that is assumed to be an analytic real function on $T^*M$. Then $p_\varepsilon$ can be written in Taylor expansion of the form 
\begin{equation}  \label{symb prin}
p_\varepsilon=p+i \varepsilon q+ \mathcal O (\varepsilon ^2 ).
\end{equation}

The spectral asymptotic theory (see Ref. \cite{Hitrik07}) states that, under some suitable global assumption of $p$ and $q$, the spectrum of  $P_{\varepsilon}$  is discrete, and included in a horizontal band of size $\mathcal{O}(\varepsilon)$.
Moreover there is a correspondence (in fact a local diffeomorphism, see Ref. \cite{QS14}), say $f$, from a small complex window of $ \mathbb C$ into $\mathbb R^2$ such that the eigenvalues of $P_{\varepsilon}$ contained in this window are sent to a part of $h \mathbb Z^2 $, modulo $\mathcal O(h^\infty)$. 
The spectrum therefore has locally the structure of a deformed lattice, with horizontal spacing $h$ and vertical spacing $\varepsilon h$.
That map $f$ admits asymptotic expansions in $h$ and $\varepsilon$, whose leading term is locally well defined.  Such a map can be seen as a \emph{local chart} for the spectral band, see Fig. 6.

\begin{figure}[!h]
	\begin{center}
		\includegraphics[width=0.8 \textwidth]{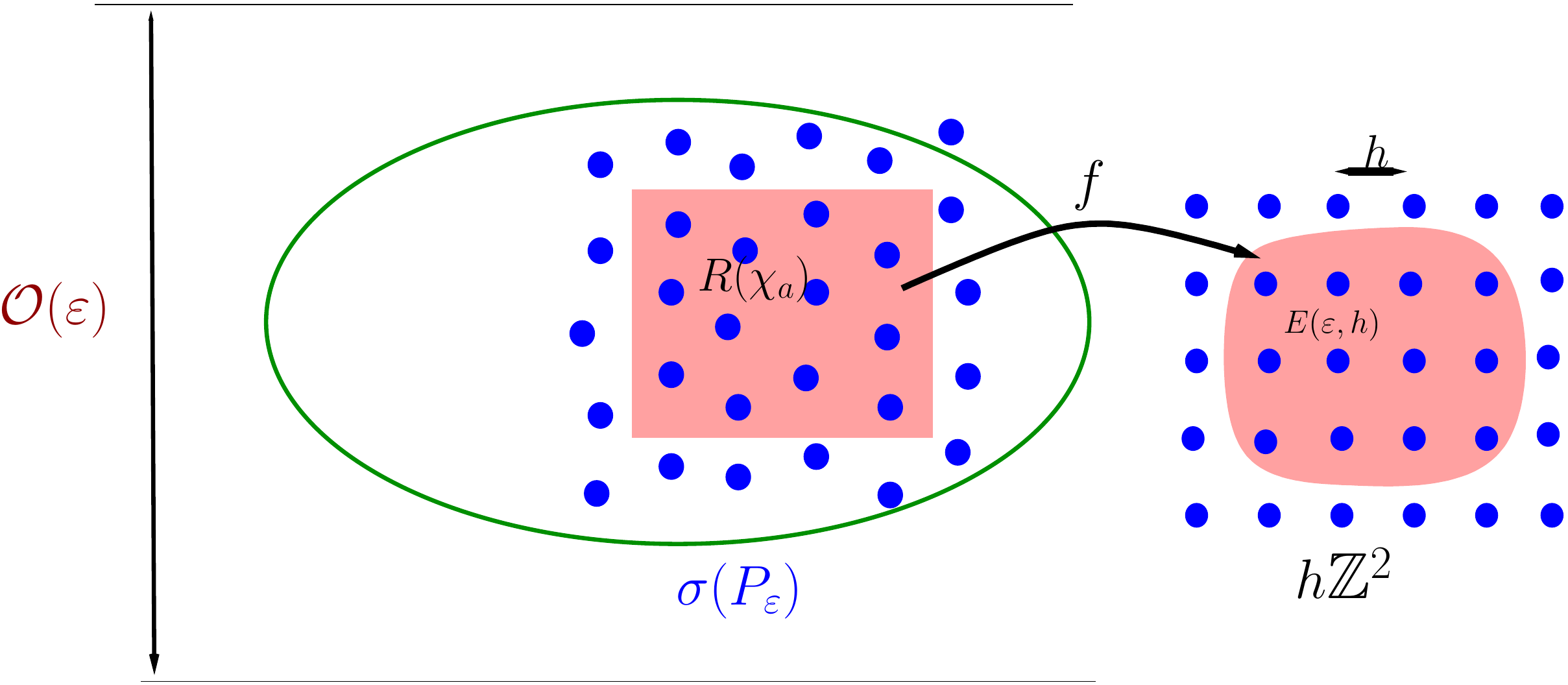}
		\caption{$h-$local chart of the spectrum}
	\end{center}
	\label{fg1}
\end{figure}

With regard to the global problem, a result of Ref. \cite{QS14} shows that the differential of the transition maps between two overlapping local charts $ (f_i, U^i(\varepsilon)) $ and $ (f_ j, U^j (\varepsilon))$ is in the group $GL(2,\mathbb Z)$, modulo $\mathcal O(\varepsilon, \frac{h}{\varepsilon})$:
$$ d (\widetilde{f}_i) = M_{ij} d( \widetilde{f}_j)+ \mathcal O(\varepsilon, \frac{h}{\varepsilon}),$$
where $\widetilde{f}_i= f_i \circ \chi$, $\widetilde{f}_j= f_j \circ \chi$, and $\chi$ is the function given by \eqref{chi}, and $M_{ij} \in GL(2, \mathbb Z)$ is an integer constant matrix.

Let $U (\varepsilon)$ be a bounded open domain in the spectral band and cover it by an arbitrary locally finite covering of local charts
$ \{ \left(  f_ j ,  U^j (\varepsilon) \right) \}_{j \in \mathcal{J} }$, here $\mathcal{J}$ is a finite index set. Then the spectral monodromy of $P_\varepsilon$ on $U (\varepsilon)$ is defined as the unique $1$-cocycle $\{  M_{ij} \} $, modulo coboundary in the first \v{C}ech cohomology group.

Moreover, due to the fact that the leading term of $\widetilde{f}_i$ in asymptotic expansions in  $\varepsilon$ and  $ \frac{h}{\varepsilon}$ can be chosen  as an integral action over invariant Liuoville tori of $p$, 
this spectral monodromy therefore allows to recover the classical monodromy of the underlying completely integrable system. More precisely, the spectral monodromy of $P_\varepsilon$ is the adjoint of the classical monodromy defined by $p$ (see Theo. 4.7, Ref. \cite{QS18}).

%%%%%%%%%%%%%%%%%%%%%%%%%%%%%%%%%%%%%%%%%%%%%%%

The above construction can be totally applied to the case of the Champagne bottle by taking $P$ as the quantization $\widehat{H}$,  given in  \eqref{qH}. We can even take the function $q$ in \eqref{symb prin} by $J$, given in \eqref{J} as a particular case.

Les us illustrate the above discussion by a numerical calculation for the spectrum of a specific operator $P_\varepsilon$. 
With the help of Mathematica, Fig. 7 shows the discrete lattice of $\sigma (P_\varepsilon)$, where we try $h=0.1$, $\varepsilon= \sqrt{h}$ and 
\begin{equation}
P_\varepsilon= \widehat{H}+ i \varepsilon\widehat{J}
+   \varepsilon ^2 x_1^5  \frac{\partial}{x_2} +h  x_2^3 \frac{\partial}{x_1}.
\end{equation}

\begin{figure}[!h]
	\begin{center}
		\includegraphics [width=0.6\textwidth]{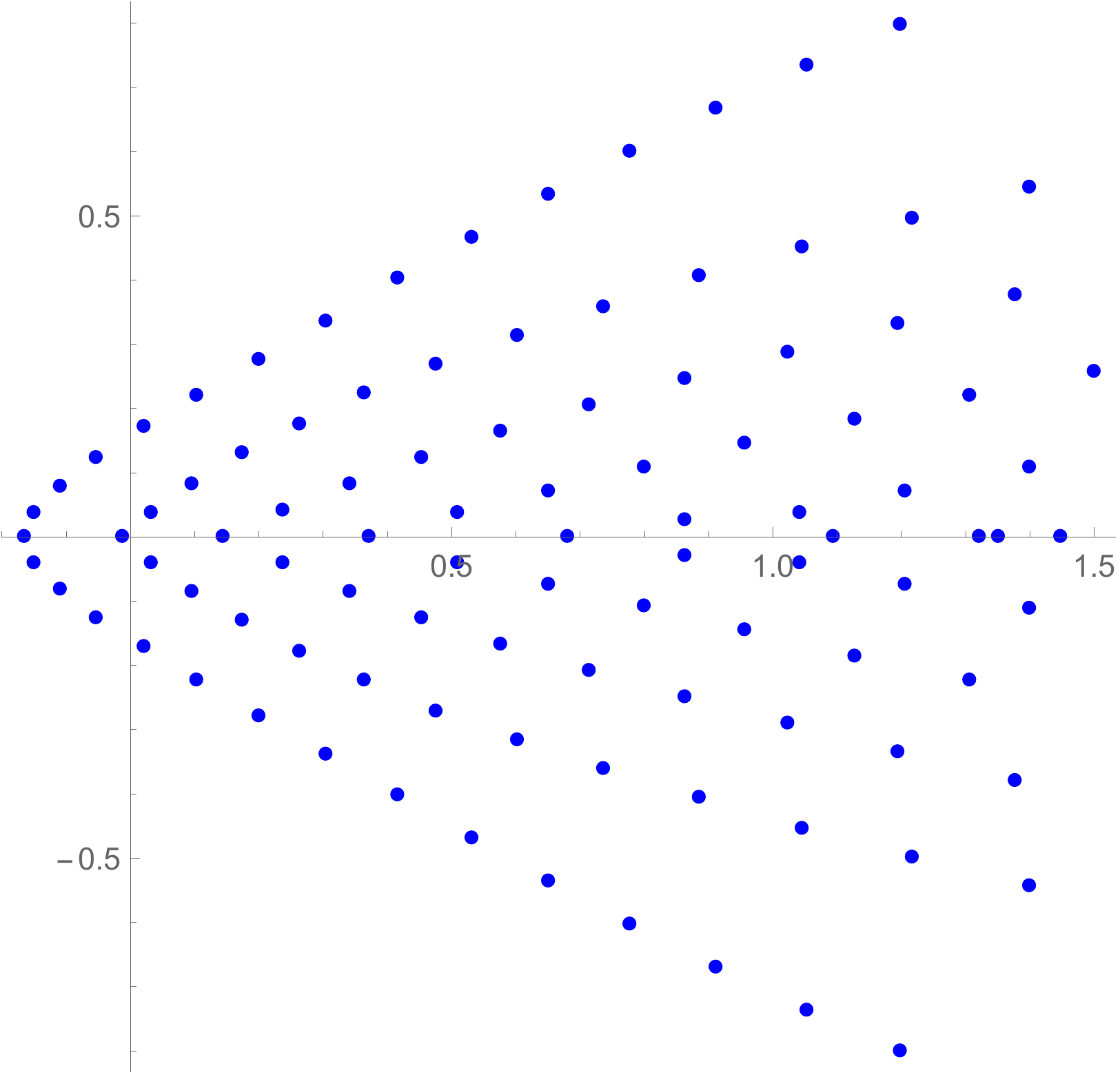}
		\caption{The spectrum of the Champagne bottle}
	\end{center}
	\label{f1}
\end{figure}

In that way the nontrivial monodromy of the Champagne bottle is well identified to the spectral monodromy of small non-selfadjoint perturbation of a selfadjoint pseudodifferential operator, whose principal symbol is the Hamiltonian of the Champagne bottle.

% \section*{Acknowledgements} 
% This work was completed during my visit to the Vietnam Institute for Advanced Study in Mathematics (VIASM). 

% I would like to thank the Vietnam Advanced Research Institute for Mathematics (VIASM) for supporting and welcoming me during my visit.

%I would like to thank the Vietnam Institute for Advanced Study in Mathematics (VIASM) for its support and hospitality during my visit in 2019. I would like also to thank the Phenikaa  University, who gave me an excellent opportunity to develop my research.

%%%%%%%%%%%%%%%%%%%%%%%%%%%%%%%%%%%%%%%%%%%%%%%%%%%%%%%%%%%%%%%%%%%%%%%%%%%%%%%%%%%%%%%%%%%%%%%%%%%%%%%%%%%%%%%%%%%%%
%%%%%%%%%%%%%%%%%%%%%%%%%%%%%%%%%%%%%%%%%%%%%%%%%%%%%%%%%%%%%%%%%%%%%%%%%%%%%%%%%%%%%%%%%%%%%%%%%%%%%%%%%%%%%%%%%%%%%%
%%%%%%%%%%%%%%%%%%%%%%%%%%%%%%%%%%%%%%%%%%%%%%%%%%%%%%%%%%%%%%%%%%%%%%%%%%%%%%%%%%%%%%%%%%%%%%%%%%%%%%%%%%%%%%%%%%%%%%%%%%

% \section*{References}

\bigskip

\end{document}